\documentclass[]{article}
\usepackage[full]{textcomp}  % Full option provides a complete rankge of symbols

\usepackage{amsmath}
\usepackage{amsthm}
\usepackage{amssymb}
\usepackage{multicol}
\usepackage{cancel}
\usepackage{comment}
\usepackage{hyperref}
\usepackage{xfrac}
\usepackage{mathrsfs}
\hypersetup{
	colorlinks=true,
	linkcolor=blue,
	citecolor=red,
	breaklinks=true
}

\usepackage[T1]{fontenc}
\usepackage[left=1.3in, right=1.3in, top=1.25in, bottom=1.25in]{geometry}
\usepackage{dirtytalk}
\usepackage[small]{titlesec}
\usepackage{mathtools}
\usepackage{comment}
\usepackage{epigraph}

\usepackage{fancyhdr}
\usepackage{natbib}
\setcitestyle{authoryear,open={(},close={)}}

\usepackage{sgame}

\ssualfalse
\gamemathtrue

\usepackage{pgfplots}
\usepackage{tikz}
\usetikzlibrary{positioning}
\pgfplotsset{width = 10cm,compat = 1.9}
\usetikzlibrary{patterns}
\usepackage{array}

\def\t{\textrm}

\newtheorem{theorem}{Theorem}
\newtheorem{lemma}{Lemma}

\newtheorem*{axiom*}{Axiom}

\newtheorem*{theorem*}{Theorem}
\theoremstyle{definition} 
\newtheorem{definition}{Definition}
\newtheorem{example}{Example}

\newcommand{\spbr}{:}

\usepackage{pifont}% http://ctan.org/pkg/pifont
\usepackage{fontawesome5}
\usepackage{hyperref}% http://ctan.org/pkg/hyperref
\usepackage[nameinlink,noabbrev,sort&compress]{cleveref}
\crefname{assumption}{assumption}{assumptions}
\crefname{example}{example}{examples}
\usepackage{algorithm}
\usepackage{algpseudocode}

\usepackage{graphicx}

\makeatletter
\newcommand{\threedots}{%
	\mathrel{\mathpalette\threedots@\relax}%
}
\newcommand{\threedots@}[2]{%
	\sbox\z@{$\m@th#1:$}%
	\vbox to\ht\z@{%
		\hbox{$\m@th#1.$}%
		\vss
		\hbox{$\m@th#1.$}%
		\vss
		\hbox{$\m@th#1.$}%
	}%
}
\makeatother
\usepackage{kpfonts} 
\usepackage{inconsolata}

\newcolumntype{C}{>{$}c<{$}} % math-mode version of "c" column type
\usepackage{ulem}
\usepackage[multiple]{footmisc}

\newcommand{\stu}{{i}}
\newcommand{\stuj}{j}
\newcommand{\Stu}{I}

\newcommand{\sch}{{s}}

\newcommand{\Sch}{{S}}
\newcommand{\quo}{q}

\newcommand{\spref}{\succ}
\newcommand{\wpref}{\succsim}

\newcommand{\swthn}{\succ^*}
\newcommand{\wwthn}{\succsim^*}

\newcommand{\sbtwn}{\succ}
\newcommand{\wbtwn}{\succsim}

\newcommand{\ebtwn}{\sim}

\newcommand{\age}{a}

\newcommand{\bge}{b}

\newcommand{\coal}{a}
\newcommand{\Coal}{A}

\newcommand{\mat}{\mu}

\newcommand{\nat}{\nu}
\newcommand{\matinv}{\mat^{-1}}

\newcommand{\natinv}{\nat^{-1}}

\newcommand{\matone}{\mat_1}
\newcommand{\mattwo}{\mat_2}

\newcommand{\matoneinv}{\matone^{-1}}
\newcommand{\mattwoinv}{\mattwo^{-1}}

\newcommand{\Upp}{U}

\newcommand{\st}{\t{ s.t. }}
\newcommand{\Topt}{T}

\title{A Unified Theory of School Choice}
\author{Peter Doe}
\date{\today}

\begin{document}
	
	\maketitle
	
	\begin{abstract}
		In school choice, policymakers consolidate a district's objectives for a school into a priority ordering over students.
		They then face a trade-off between respecting these priorities and assigning students to more-preferred schools.
		However, because priorities are the amalgamation of multiple policy goals, some may be more flexible than others.
		This paper introduces a model that distinguishes between two types of priority: a between-group priority that ranks groups of students and must be respected, and a within-group priority for efficiently allocating seats within each group.
		The solution I introduce, the unified core, integrates both types.
		I provide a two-stage algorithm, the DA-TTC, that implements the unified core and generalizes both the Deferred Acceptance and Top Trading Cycles algorithms.
		This approach provides a method for improving efficiency in school choice while honoring policymakers' objectives.
	\end{abstract}
	
	\section{Introduction}\label{sec:intro}
	School choice programs\footnote{What economists refer to as \textit{school choice}---the ability to choose between public schools---is frequently called \say{open enrollment} in the media. Although, in common parlance, \say{school choice} refers to voucher or tax credit programs for defraying the costs of private schooling or homeschooling, I will use the standard language in economics.} allocate school seats using \textit{priorities} (such as sibling status or proximity): when a school is oversubscribed, these priorities determine which students are admitted.
	The interpretation of these priorities is fundamental to the program's design, and there are two leading interpretations.
	The stronger interpretation posits that priorities are the right to attend the school ahead of those with lower priorities.
	In the weaker interpretation, priorities reflect \say{better opportunities} to attend a school, all else equal \citep{abdulkadiroglu_school_2003}.
	
	These interpretations underpin the central trade-off in school choice between fairness and efficiency.
	\textit{Fairness} means no student is denied admission to a school in favor of a lower-priority student; \textit{efficiency} means students are assigned to schools in such a way that no student can be improved without harming another student.
	By requiring the stronger interpretation, the designer arrives at a {fair} match because no student prefers a school that a lower-priority student is assigned to.
	However, allowing for the weaker interpretation shifts the set of permissible assignments to include an {efficient} match.
	The main dilemma is that no algorithm delivers a match that is both fair and efficient.
	
%	The tension comes from how the two interpretations assign \say{rights.}
%	In the weaker interpretation, priorities translate into a form of ownership rights: the highest-priority student owns the school and can trade it with other students.
%	By contrast, the stronger interpretation gives students no rights over the schools; each school is treated as an agent who decides with whom it is willing to give its seats to.
%	Allocating rights to the students produces efficiency, but allocating them equally between students and schools guarantees fairness.
	
	Despite the appeal of efficiency, policymakers consistently choose fairness.\footnote{New Orleans Recovery School District and the Common App is the only example I know of that sought an efficient outcome in some rounds of the assignment process, but this was abandoned after the first year.}
	Empirically, the cost is substantial: for instance, \cite{abdulkadiroglu_strategy-proofness_2009} finds that the assignments of 4300 eighth-grade students in New York City could have been improved without harming any other students.
	In response to this, recent research focuses on weaker versions of fairness to allow for efficiency gains. 
	
	However, an examination of why policymakers side with fairness reveals that they are frequently concerned with the violations of some priorities but not with others.
	For instance, a priority derived from a high test score may have a stronger interpretation than a priority assigned based on proximity to the school.
	Because economists have not considered how policymakers may treat some priorities differently than others, this omission has precluded a promising approach to reconciling fairness and efficiency.
	The key to understanding how policymakers view priorities lies in the two-step process that is typically used to construct priorities.
	
	In the first step, the policymaker identifies a set of student characteristics to prioritize and uses these characteristics to partition students into priority groups.\footnote{For example, these characteristics frequently include whether the student has a sibling attending the school, whether the student is within a particular geographic region, the student's test scores, etc., and a priority group consists of the students with the same set of characteristics.}
	Then, the policymaker assigns a priority over the priority groups that reflects the policymaker's objectives for the school; this is the \textit{between-group} priority.
	In the second step, within each group, another \textit{within-group} priority is provided to break ties, which could depend upon other student characteristics or a random lottery. 
	A student's priority at a school is lexicographically determined first by her between-group priority and then by her within-group priority.
	
	Policymakers understand that because priorities arise from different sources, they require different interpretations.
	Allocating a seat to a student in a lower priority group ahead of a student in a higher priority group may be unacceptable, but assignments within groups may be more subtle.
	The prominent case of Boston Public Schools (BPS)---the first district to carefully consider its assignment mechanism---illustrates this difference.
	
	BPS created two priority groups: students with a sibling attending the school and all other students, with the former given higher between-group priority than the latter.
	Within each priority group, BPS gave higher within-group priority to students residing in the school's walk zone, and broke any remaining ties with a random lottery.
	When comparing the efficient algorithm with the fair algorithm, Superintendent Payzant implies that the between-priority derived from having a sibling present should be treated differently from the within-priorities:
	\begin{quote}
		[The efficient system] presents the opportunity for the priority of one student at a given school to be \say{traded} for the priority of a student at another school, [\ldots] There may be advantages to this approach, [\ldots] It may be argued, however, that certain priorities -- e.g., sibling priority -- apply only to students for particular schools and should not be traded away. \citep{abdulkadiroglu_changing_2006}
	\end{quote}
	The implication from his silence concerning walk-zone priority is that, in some cases, walk zone priorities \textit{could} be traded away.
	Sibling priority has the stronger interpretation while walk zone priority has the weaker interpretation of a better opportunity to attend the school.\footnote{Readers familiar with this story will note that my exposition of the priority groups is slightly different than that of \cite{abdulkadiroglu_changing_2006}. Formally, BPS designated five priority groups: continuing students, sibling-walk, sibling, walk, and all others. However, as Payzant's quote illustrates, sibling priority and walk zone priority are normatively different. Conveniently, the order on the priority groups allows for walk zone priority to be viewed as a tiebreaker within the sibling priority group. If BPS ranked these groups differently, my model would not apply directly. For ease of exposition, I exclude the continuing students from the model; continuing students can be easily added to the model by including an additional priority group.}
	
	Despite the different interpretation of priorities within these two steps, economists usually remain agnostic about the source of the priority.
	The nuances between priority groups are dismissed, and all priorities are treated equally.
	The existing options afforded to policymakers do not allow them to express these nuances.
	This leaves policymakers in a quandary because a more efficient match is better than not, but violating the stronger priorities may be impermissible.
	When BPS faced this decision, they erred on the side of respecting priorities and implemented the fair mechanism even though only a subset of their priorities required the stronger interpretation.
	
	In this paper, I introduce a model that explicitly distinguishes between two types of priority---between-group and within-group---and I propose the unified core to connect the two.
%	I present a middle road that allows policymakers to specify which priorities require the stronger interpretation and which do not.
	In my model, each school's priority consists of two layers.
	The first layer is the between-group priority, a weak preference that specifies which students belong to which priority group and how those groups are ordered.
	The second layer is the within-group priority, a strict preference that refines the between-group priority.\footnote{In most applications, the within-group priority consists of a random lottery. However, I refrain from referring to it as a \textit{tiebreaker} because the within-group priority may include substantive components, such as geographic location (as in the case of BPS). The emphasis, however, is that these priorities have a weaker interpretation.}
	Accordingly, only the between-group priority is interpreted as a right to attend the school ahead of lower-priority students, and it cannot be violated.
	On the other hand, the within-group priority has the milder interpretation of an opportunity to attend the school; it is a means of allocating the school seat, but it is not inviolable.
	The interpretation of the within-group priority is weak enough to guarantee efficiency within a group while still using the priority to allocate the seat.
	
	The unified core connects the two cores that underlie the fair and efficient algorithms, and is implemented by carefully combining the algorithms that implement the fair and efficient cores.
	The core underlying the fair algorithm is the \textit{fair} core of \cite{gale_college_1962} (henceforth GS), which allows a student to claim a seat at a school if she has a higher priority than a student currently there.
	This stems from the strong interpretation of priorities and treats the school as a strategic player and the priority as the school's preference.
	The result is a fair match because any unfair match can be blocked by the student whose priority is violated.
	
	The core underlying the efficient algorithm is from \cite{rong_core_2022} (henceforth RTZ),\footnote{It is worth mentioning here that RTZ's core identifies precisely the \textit{just} assignments of \cite{morrill_making_2015}.} who develop their \textit{efficient} core to provide a foundation for the use of several efficient algorithms in school choice.
	RTZ's efficient core allows a student to take possession of a school provided no higher-priority students can veto that action.
%	This is implemented by allowing a coalition of students to enforce a match between a student and a school only if there are more empty seats at the school than there are higher-priority students outside of the coalition. 
	In effect, the efficient core allows for a form of tradable priorities, where a low-priority student can receive a priority from another student so long as no higher-priority students are harmed.
	As is paralleled in many economic models, the power to trade freely (in this case, trade school seats) guarantees that an efficient allocation is reached.
	
	Whether policymakers prefer the outcomes of the fair core or the efficient core ultimately depends on the source of the students' priorities.
	If the priority reflects different priority groups, then the fair core is appealing; if it reflects within-group tiebreakers such as a random lottery, the efficient core allows more students to attend the schools that they prefer.
	However, neither the fair core nor the efficient core allows for combinations of the two priorities.
	
	The unified core resolves this tension and connects the two cores by allowing the policymaker to explicitly designate two kinds of priorities.
	By treating the between-group priorities as in the fair core while the within-group priorities as in the efficient core, the unified core generalizes both cores.
	Students can always claim a seat at a school if they have a higher between-group priority than a currently assigned student, as in GS.
	Within-group priorities, however, may be traded amongst students within the same priority class, as in RTZ.
	This guarantees a form of \say{within-group} efficiency.
	The way I implement this mirrors RTZ: a student can take a seat provided no within-group interrupters in her priority group can block her.
	
	The main challenge in crafting a solution for this model is integrating the differing normative implications of the two priorities.
	Unlike the model of \cite{erdil_whats_2008} (henceforth EE),\footnote{I discuss the differences in models in \Cref{sub:lit}, and the differences in our techniques in \Cref{sec:disc}. In short, EE uses only between-group priorities; however, I must use more care because of the within-group priorities. This prevents a direct application of the Stable Improvement Cycles algorithm.} the within-group priorities are not merely random tiebreakers.
	Within-group priority may encode substantive differences (as in the case of BPS) that must be respected, just not in the same fashion as the between-group priority.
	The difficulty lies in creating a unified model that allows for both interpretations, one between groups and the other within groups.
	
	The power of the unified core lies in its ability to interpolate between the frameworks of GS and RTZ.
	When the between-group priority is a strict ranking where each student forms a singleton priority group, the unified core corresponds to the fair core of GS.
	On the other hand, when the between-group priority is indifferent over all students (there is a single priority group), then the unified core corresponds to the efficient core of RTZ.
	At intermediate stages, the unified core blends both RTZ and GS in a principled manner. 
	
	To make the unified core practical, I introduce a novel two stage algorithm---the DA-TTC algorithm---that always produces a match in the unified core.
	In the first stage, I allow students to freely apply to their most-preferred schools as in the Deferred Acceptance algorithm of GS.
	Schools tentatively accept the highest-priority applicants and immediately reject the rest.
	The process continues until no more students are rejected, at which point the first stage terminates and returns the match $\matone$.
	
	In the second stage, I adapt Gale's Top Trading Cycles algorithm but with a restriction on the allowable trades \citep{shapley_cores_1974}.
	Essentially, students may only trade their seats (that they have from $\matone$) to other students in the same priority group.
	I iteratively find groups of students who can, by trading their assigned schools among themselves, be matched to their most-preferred school out of the ones currently available.
	This algorithm is the crux of this paper, and is presented in \Cref{sec:subTTC}.
	I label the resulting match $\mattwo$.
%	Operationally, because $\matone$ is the result of the Deferred Acceptance, this means that only students in the lowest priority group at each school are \textit{active}.
%	In each round, each school points to its most-preferred active student assigned to it by $\matone$.
%	Each active student then points to her most-preferred school among those that she is in the same priority group as the student the school points to.
%	A cycle must exist, and each student on that cycle is permanently assigned to the school she points to and then is no longer active.
%	With the cycle removed, the process repeats itself.
	Once all students have been removed, the second stage terminates and returns the match $\mattwo$.
	
	My main result establishes that $\mattwo$ belongs to the unified core.
	The advantage of my approach is a practical algorithm to finding an element of the unified core.
	The DA-TTC is also computationally feasible.\footnote{It runs in polynomial time.}
	
	A key advantage of the unified core in applications is the simplicity of its definition.\footnote{Indeed, Payzant expressed reservation concerning the TTC because of the complexity of the algorithm. I return to this topic in the discussion.}
	First, there are no between-group priority violations, so the fairness property of the fair core is preserved.
	However, when there is a within-group priority violation, the explanation is straightforward: rectifying the violation would harm a student with higher within-group priority.
	By focusing on a reasonable outcome, the definition of the unified core lends itself to applications in school choice.
	
	The rest of this paper is structured as follows.
	First, I close out this section with a discussion of the relevant papers, emphasizing the differences between the unified core, related concepts, and the methods used.
	\Cref{sec:mod} introduces the model, the bulk of which is devoted to the distinction between within-group and between-group priorities.
	In \Cref{sec:analysis}, I develop my main result (\Cref{thm:wthn-block}) which shows that a match in the unified core can always be found using an efficient algorithm.
	I conclude in \Cref{sec:disc} with a discussion of the result and further avenues of research.
	
	\subsection{Related Literature}\label{sub:lit}
	This paper connects several strands of research.
	
	First, as mentioned previously, the model in this paper has substantial overlap with \cite{erdil_whats_2008} (henceforth, EE).
	The key difference between my model and EE's model is that I include the within-group priorities in the model primitives, while they take these priorities as randomly generated from the between-group priorities.
	Although mathematically similar, because the within-group priorities may carry normative implications, the outcomes in the unified core are restricted.
	In EE, any student within the same priority group can be assigned to the school so long as the match is Pareto optimal among the matches that are stable with respect to the between-group priorities.
	In contrast, in the unified core, a student of higher within-group priority can claim a seat so long as no higher priority students in the same group are harmed.
	The within-group priority disciplines how seats are allocated within priority groups.
		
	Second, ever since \cite{abdulkadiroglu_school_2003}, researchers have been trying to modify the Deferred Acceptance algorithm to accommodate efficiency gains \citep{ehlers_illegal_2020,kesten_school_2010,troyan_essentially_2020,reny_efficient_2022}.
	These approaches, however, take the priorities as given and weaken the definition of the core.
	The unified core complements this approach by weakening the priorities and then adapting the core to this framework.
	The upshot of the unified core is that the definition of the core is simpler, but the downside is that the preferences must have a two-stage structure to make any efficiency gains.
	
	Third, a separate approach has been to provide foundations for the use of the Top Trading Cycles algorithm in school choice \citep{abdulkadiroglu_role_2010, morrill_alternative_2013, abdulkadiroglu_minimizing_2017,chen_new_2021,rong_core_2022, dur_characterization_2024}.
	These axiomatic definitions ground this work, and I extend our understanding of these mechanisms by integrating them into the fairness framework of GS.
	
	The paper closest in spirit to mine is \cite{abdulkadirog_generalized_2011}.
	The main difference between his paper and mine is the restriction they place on between-priorities: in their model, between-group priorities either express indifference or a strict preference over all students.
	Additionally, his focus is on student optimality rather than on the core.
	He uses the within-priorities in a similar way, but only in his algorithm; the unified core provides the underpinnings for his solution.
	
	\section{Model}\label{sec:mod}
	In this section I present the formal model of a \textit{school choice problem} and my key definitions.
	In the first subsection, I introduce the standard primitives, except that I replace a school's one priority with two priorities: the between-group priority and within-group priority.
	In the second subsection, I develop the theory of the unified core.
	In the third subsection, I turn to the solutions of GS, RTZ, and EE, and highlight the differences and advantages of the unified core.
	
	\subsection{The Primitives}
	There is a finite set of \textit{students} $\Stu$ and a finite set of \textit{schools} $\Sch$, collectively called \textit{agents}.
	Each school has a \textit{capacity} $\quo_\sch \geq 1$ of seats.
	A \textit{match} is a function $\mat:\Stu \rightarrow \Stu \cup \Sch$ with the following two properties:
	\begin{enumerate}
		\item if $\stu \in \Stu$, then $\mat(\stu) \in \Sch \cup \{\stu\}$; and
		\item if $\sch \in \Sch$, then $|\matinv(\sch)| \leq \quo_\sch$.
	\end{enumerate}
	The first requirement guarantees that the market is two-sided: students are matched to schools or unmatched (this is denoted as being matched to oneself).
	The second requirement guarantees that capacities are not exceeded.
	
	Every student has a strict preference $\spref_\stu$ over $\Sch \cup \{\stu\}$; the associated weak preference is $\wpref_\stu$.
	A match $\mat$ is \textit{individually rational} if for every $\stu \in \Stu$, $\mat(\stu) \wpref_\stu \stu$.
	When considering a set of students $\Coal \subseteq \Stu$, I say that $\Coal$ \textit{prefers} $\nat$ to $\mat$ if for every $\stu\in \Coal$, $\nat(\stu) \wpref_\stu \mat(\stu)$ and the comparison is strict for some student in $\Coal$.
	
	Schools, however, come equipped with two orders over $\Stu$, the between-group priority and the within-group priority.\footnote{Implicit here is the assumption that every student is acceptable to every school. Extending this model to incorporate unacceptable students is straightforward, but adds unnecessary notation.}
	The \textit{between-group priority} order is a weak preference $\wbtwn_\sch$.\footnote{The irreflexive portion is $\sbtwn_\sch$ (the strict preference relation) and the reflexive portion is $\ebtwn_\sch$ (the indifference relation).}
	It represents the priorities that require the stronger interpretation as a right to attend the school ahead of others.
	When $\wbtwn_\sch$ does not rank two students strictly, it means that they are in the same priority group.
	Formally, denote the \textit{priority group} of student $\stu$ at $\sch$ as
	\begin{align*}
		[\age]_\sch \equiv \{\stuj \in \Stu \spbr \stuj \ebtwn_\sch \stu\}
	\end{align*}
%	Borrowing notation from algebra, I denote the priority group of $\age$ at $\sch$ as $[\age]_\sch$.
	
	The \textit{within-group priority} order is a strict preference $\swthn_\sch$.
	It represents the priorities that are allowed the weaker interpretation of a better opportunity to attend the school.
	In practice, the priority that a policymaker uses is the lexicographic combination of $\wbtwn_\sch$ and $\swthn_\sch$: first, $\wbtwn_\sch$ is used to partition students into groups, and then $\swthn_\sch$ is used to break ties within groups.
	Rather than define a {combined} priority, I assume that $\swthn_\sch$ is a {refinement} of $\wbtwn_\sch$; that is, if $\age \sbtwn_\sch \bge$, then $\age \swthn_\sch \bge$.
	That is, $\swthn_\sch$ is what most researchers refer to as school priority, whereas for me it is the result of the policymaker breaking ties within priority groups.
	
	Because this model is many-to-one, it is necessary to construct the school's between-group priority over groups of students.\footnote{Extending the within-group priority in a similar manner is unnecessary in my analysis.}
	Toward this end, I make the standard assumption that the school's between-group priority over groups of students is \textit{responsive} \citep{roth_college_1985}.
	In words, this means that for equally-sized groups, the first has greater between-group priority than the second if every student in the first can be paired with a student in the second such that the first has greater between-group priority.
	For unevenly sized groups, the smaller group never has greater between-group priority than the larger group, but the larger group has greater between-group priority if there is a subset of the larger group that is the same size as the smaller group and has greater between-group priority.
	
	Formally, I construct $\wbtwn_\sch$ for groups of students in the following way.
	For $\Coal, \Coal' \subseteq \Stu$ such that $|\Coal| = |\Coal ' |$, I write $\Coal \wbtwn_\sch \Coal ' $ if the students in $\Coal$ and $\Coal' $ can be indexed such that $\coal_1 \wbtwn_\sch \coal_1'$, $\coal_2 \wbtwn_\sch \coal_2'$, \ldots, $\coal_{|\Coal|} \wbtwn_\sch \coal_{|\Coal|}'$.
	When $\Coal \wbtwn_\sch \Coal ' $ and $\Coal ' \wbtwn_\sch \Coal $, I write $\Coal \ebtwn_\sch \Coal ' $.\footnote{Note that $\wbtwn_\sch$ may not be a complete relation over groups of students: there are some coalitions $\Coal$ and $\Coal ' $ such that neither $\Coal \wbtwn_\sch \Coal ' $ nor $\Coal ' \wbtwn_\sch \Coal $.}
	If $\Coal \wbtwn_\sch \Coal ' $ but not $\Coal ' \wbtwn_\sch \Coal $, then I write $\Coal \sbtwn_\sch \Coal ' $.
	When $|\Coal| \neq |\Coal ' |$, I write $\Coal \wbtwn_\sch \Coal ' $ if $|\Coal| > |\Coal '| $ and there is some $\tilde\Coal \subseteq \Coal '$ such that $\tilde\Coal \wbtwn_\sch \Coal ' $.

	\subsection{Unified Core}
	With the primitives in hand, I turn to a discussion of the core of the model.
	As introduced earlier, policymakers view between-group and within-group priorities differently.
	To accommodate both types of priorities, I introduce two separate types of blocks.
	In both cases the conditions on student preferences are the same: the coalition of students must prefer the new match to the old one.
	What changes is the definition of what the coalition can enforce.
	
	The between-group priorities resemble the preferences in the GS two-sided model.
	However, within my framework I focus on students as the active players.
	To accommodate this, \Cref{def:btwn-enfo} allows a student to block a match if her between-group priority has been violated.
	This definition mirrors the standard blocking definition of GS.
	It rules out any violations of the between-group priority order.
	\begin{definition}\label{def:btwn-enfo}
		The coalition $\Coal$ can \textit{between-group enforce} match $\nat$ over match $\mat$ if for every $\stu\in \Coal$, either there is a student $\stuj\in \matinv(\nat(\stu))$ such that $\stu \sbtwn_\sch \stuj$ or $|\matinv(\nat(\stu))| < \quo_{\nat(\stu)}$.\footnote{Traditionally in GS-style models, the schools are included in the coalition; however, in RTZ-style models (see next paragraph), this is not the case. I present this student-only version for use here. The downside is that the definition of enforcement relies on $\mat$; the upside is fewer cases in the definition of the unified core.}
	\end{definition}
	The definition of blocking is standard: a coalition can \textit{between-group block} match $\mat$ if it can between-group enforce some match $\nat$ that the coalition prefers to $\mat$.
	Note that if a match is not individually rational, then it is between-group blocked by some student who wishes to be unmatched.
	
	Between-group blocks are useful in school choice models, but do not fully model the priorities given by schools.
	As seen in EE, using between-group blocks alone can provide substantial welfare improvements for students.
	However, the within-priorities may also carry normative value as Payzant implied.
	To address this, I turn to a treatment of within-group priorities.
	
	The within-group priorities are modeled as in the RTZ framework.
	In their model, a coalition of students can {enforce} a match to some schools only if no group of students \textit{outside} the coalition could veto such a match by asserting their stronger claims.
	I import this definition into my model by restricting attention only to students within the same priority group.
	To formalize this, I first define the \textit{within-group upper contour set}:
	\begin{align*}
		\Upp_\sch(\stu) &= \{\stuj \in [\stu]_\sch \spbr \stuj \wwthn_\sch \stu\}
	\end{align*}
	In words, $\stuj$ is in $\Upp_\sch(\stu)$ if $\stuj$ has a higher within-group priority at $\sch$.
	When considering whether $\stu$ can claim a seat at $\sch$, $\Upp_\sch(\stu)$ are the students who could prevent this match because they have a stronger within-group claim to $\sch$.
	For example, suppose $\quo_\sch = 1$; if $\stu$ wishes to match to $\sch$, then $\stu$ must guarantee that students in $\Upp_\sch(\stu)$ are not harmed by this action.
	I call students in $\Upp_\sch(\stu)$ but not in $\Coal$ the \textit{within-group interrupters} of $\stu$.
	A sufficient condition is that there are no within-group interrupters.
	When $\quo_\sch > 1$, it is not necessary to eliminate every within-group interrupter: if $n$ students from $\Coal$ match to $\sch$, then $\quo_\sch - n$ within-group interrupters do not need to be included because the within-group interrupters (even when acting together) are unable to prevent the coalition members from claiming their seats at the school.\footnote{The reader may wonder why I (and RTZ) include this slackness in the definition of enforceability. Without it, a student with top priority might trade \textit{multiple} seats away. The intuition is that, if every within-group interrupter is included, then including the highest-priority student might necessitate (through a chain of enforceability claims) including two or more lower-priority students who are matched to the school. An example demonstrating this is available upon request.
%		For instance, consider the following example with four students $\{\}$ and three schools:
%		\begin{center}
%			\centering
%			\begin{multicols}{3}
%				\begin{tabular}{C | C | C  | C}
%					\spref_1 & \spref_2 & \spref_3 & \spref_4\\
%					\hline
%					a & b & b & b \\
%					b & a & a & a \\
%					c & c & c & c
%				\end{tabular}
%				\\
%				\begin{tabular}{C | C | C }
%					\swthn_a & \swthn_b & \swthn_c \\
%					\hline
%					1 & 1  & 3 \\
%					3 & 2 &  1 \\
%					1 & 3 &  2 \\
%					1 & 3 &  2 \\
%				\end{tabular}
%			\end{multicols}
%		\end{center}
	}
	\Cref{def:wthn-enfo} formalizes this intuition.
	\begin{definition}\label{def:wthn-enfo}
		A coalition $\Coal \subseteq \Stu$ can \textit{within-group enforce} match $\nat$ over match $\mat$ if for every school $\sch \in \nat(\Coal)$, the following two conditions hold:
		\begin{enumerate}
			\item $\natinv(\sch) \ebtwn_\sch \matinv(\sch)$
			\item the sum of the number of within-group interrupters across all students in $\Coal$ that are matched to $\sch$ is less than or equal to $\quo_\sch - |\natinv(\sch)|$.
		\end{enumerate}
	\end{definition}
	Again, the definition of blocking is symmetric: a coalition can \textit{within-group block} match $\mat$ if it can within-group enforce some match $\nat$ that the coalition prefers to $\mat$.
	
	The main difference between within-group blocks and between-group blocks is their treatment of schools.
	Schools must strictly benefit (according to the between-priority) from between-group blocks, but may be indifferent in a within-group block.
	The within-group block compensates by placing a stronger condition upon the student coalition.
	
	With both types of blocks in hand, I turn to defining the unified core in \cref{def:unicore}.
	\begin{definition}\label{def:unicore}
		$\mat$ is in the \textit{unified core} if it is neither between-group blocked nor within-group blocked.
	\end{definition}
	
	In \Cref{sec:analysis}, I provide a two-stage algorithm that always finds a match in the unified core.
	Before introducing the algorithm, I first discuss the two most-closely related solutions.
	
	\subsection{Comparison with Related Models}
	In this subsection, I briefly introduce the models of GS and RTZ.

	In the fairness framework of GS, a match $\mat$ is \textit{fair} if there is no student $\stu$ and $\sch$ such that $\sch \spref_\stu \mat(\stu)$ and either there is a $\stuj \in \matinv(\sch)$ such that $\stu \swthn_\sch \stuj$ or $|\matinv(\sch)| < \quo_\sch$.
	Note that when the between-group priority orders all students strictly, then the set of fair matches and the unified core are the same.
	
	In contrast, the efficiency framework of RTZ uses priorities to allocate school seats as if they were objects.
	The within-group priority represents a form of ownership of the school and allows for a form of trading.
	The key difference between the efficient core and the unified core is that the efficient core assumes that there is a single priority group: any group of students can use their priorities to trade their schools.
	This is made by a slight change in definitions: the \textit{efficient upper contour set} is 
	\begin{align*}
		\Upp^*_\sch(\stu) = \{\stuj \in \Stu \spbr \stuj \swthn_\sch \stu \}
	\end{align*}
	and for a coalition $\Coal$ and match $\nat$, the \textit{efficient interrupters} of $\stu \in \Coal$ are the students in $\Upp^*_{\nat(\stu)}(\stu)$ but not $\Coal$.
	The definition of \textit{efficient enforcement} is then the same as within-group enforcement except that the condition $\natinv(\sch) \ebtwn_\sch \matinv(\sch)$ is dropped and the efficient interrupters are counted instead of within-group interrupters.
	The definition of efficient blocking and the efficient core is similar.
	When the between-group priority places all students within the same priority group, the efficient core and the unified core are the same.
	
	The power of the unified core is its ability to interpolate between the efficient and fair cores.
	By allowing for the between-group priority to range between a strict ranking of all students and complete indifference, the unified core captures the nuances of policymakers' objectives and the two-stage nature of priorities.
	In \cref{exa:interpolate}, I show how the unified core models what the fair and efficient cores cannot.
	
	\begin{example}\label{exa:interpolate}
		Consider a school choice problem with six students and six schools, each school having one seat.
		The students' preferences and the schools' between- and within-priority orders are as follows (unlisted preferences/priorities are irrelevant):
		\begin{center}
			\begin{tabular}{C | C | C | C | C | C}
				\spref_1 & \spref_2 & \spref_3 & \spref_{1'} & \spref_{2'} & \spref_{3'} \\
				\hline
				a & b & a & a' & b' & a'\\
				b & a & c  & b' & a' & c' \\
				\threedots & \threedots & \threedots  &\threedots &\threedots & \threedots\\
			\end{tabular}
			\\ \vspace{5mm}
			\begin{multicols}{2}
				\begin{tabular}{C | C | C | C | C  | C }
					\sbtwn_a & \sbtwn_b & \sbtwn_c& \sbtwn_{a'}& \sbtwn_{b'} & \sbtwn_{c'} \\
					\hline
					1,2,3 & 1,2  & 3 & 2 ' & 1 ' & 3 ' \\
					\threedots & \threedots &  \threedots & 3' & 2' & \threedots \\
					&  &   & 1 '  &  \threedots & \\
					&  &   & \threedots  &  & \\
				\end{tabular}
				\\
				\begin{tabular}{C | C | C | C | C  | C }
					\swthn_a & \swthn_b & \swthn_c& \swthn_{a'}& \swthn_{b'} & \swthn_{c'} \\
					\hline
					2 & 1  & 3 & 2 ' & 1 ' & 3 ' \\
					3 & 2 &  \threedots & 3' & 2' & \threedots \\
					1 & \threedots &   & 1 '  &  \threedots & \\
					\threedots&  &   & \threedots  &  & \\
				\end{tabular}
			\end{multicols}
		\end{center}
		Notice that this example consists of two parallel problems side-by-side.
		In the first, the students $1$ and $2$ are the highest-priority students at the school the other most-prefers (schools $a$ and $b$, respectively), and in the second, students $1'$ and $2'$ are the highest-priority at the school the other most-prefers (schools $a'$ and $b'$, respectively).
		The fair core is the match:
		\begin{align*}
			\mat^{\t{GS}}&= 
			\begin{pmatrix}
				1 & 2 & 3  & 1' & 2' &3'\\
				b & a  & c & b'& a'& c'
			\end{pmatrix}
		\end{align*}
		The fair core does not incorporate the difference between the between-group and within-group priorities; hence it is less efficient than allowed.
		However, the efficient core is also unsatisfactory:
		\begin{align*}
			\mat^{\t{RTZ}}&= 
			\begin{pmatrix}
				1 & 2 & 3  & 1' & 2' &3'\\
				a & b  & c & a'& b'& c'
			\end{pmatrix}
		\end{align*}
		The difficulty with RTZ is that it does not recognize that student $3'$ can claim a seat at school $a'$ ahead of student $1'$.
		
		The unified core, however, correctly identifies the difference between these priority groups:
		\begin{align*}
			\mat^{\t{UC}}&= 
			\begin{pmatrix}
				1 & 2 & 3  & 1' & 2' &3'\\
				a & b  & c & b'& a'& c'
			\end{pmatrix}
		\end{align*}
	\end{example}

	\section{Analysis}\label{sec:analysis}
	In this section, I present my main result (\Cref{thm:wthn-block}).
	It states that the unified core is non-empty for every school choice problem and provides an algorithm for finding one such match.
	
	I use a two-stage algorithm to prove the result.
	In the first stage, I leverage the Deferred Acceptance algorithm (DA) to produce a match that is not between-group blocked \citep{gale_college_1962}.
	In the second stage, I build a variation of the Top Trading Cycles algorithm (TTC) to remove any within-group blocks from the previous match while not adding any between-group blocks \citep{shapley_cores_1974}.
	The match produced is in the unified core.
	
	\subsection{Removing between-group blocks: The Deferred Acceptance Algorithm}
	In this subsection, I present the first stage of the DA-TTC: the {Deferred Acceptance} algorithm from GS.
	In the DA, students initially start unmatched.
	In the first round, students \say{propose} to their favorite school.
	Every school tentatively accepts students up to its capacity (picking the highest according to its within-group priority) and immediately rejects the rest.
	In subsequent rounds, students who are not tentatively accepted propose to their favorite school among those that have not rejected them.
	Every school considers the new proposals simultaneously with the students it has tentatively accepted, and again rejects all but the top students up to its capacity.
	This process continues until every student is matched or has been rejected by every acceptable school.
	I formally write this in \Cref{alg:DA}.
	
	\begin{algorithm}
		\caption{Deferred Acceptance Algorithm}\label{alg:DA}
		\begin{algorithmic}
			\State every student $\stu$ points to the $\wpref_\stu$-highest agent
			\While{more than $\quo_\sch$ students point to some school $\sch$}
			\State \quad  $\sch$ rejects all but the $\swthn_\sch$-top $\quo_\sch$ students who proposed to $\sch$
			\State \quad every student $\stu$ points to the $\wpref_\stu$-highest agent \textit{who has not rejected him yet}
			\EndWhile
			\State $\matone$ assigns each student to the agent he last proposed to
		\end{algorithmic}
	\end{algorithm}
	
	The DA was originally designed by GS to construct a fair match.
	Fairness is a stronger requirement than not being between-group blocked, this result applies to my model.
	\Cref{lem:btwn-block} translates the standard result into my model.
	The proof is standard and is relegated to the Appendix.
	\begin{lemma}\label{lem:btwn-block}
		$\matone$ is not between-group blocked by any coalition.
	\end{lemma}
	
	\subsection{Removing within-group blocks: The Top Trading Cycles Algorithm}\label{sec:subTTC}
	In this subsection, I present the second stage of the DA-TTC: the {Top Trading Cycles algorithm} of \cite{shapley_cores_1974}.
	In the TTC, each student \say{owns} a seat at the school she is matched to in $\matone$.
	Students are allowed to trade the seats that they own with others, but only within priority groups.
	
	The way that I implement this is by first restricting which students and schools are active, and then restrict which schools a student can \say{point} to.
	At the start of the TTC, each student $\stu$ is  \textit{active} if she is in the lowest priority group of $\matone(\stu)$ among the students matched to $\matone(\stu)$ (i.e.\ for every $\stuj \in \matoneinv(\matone(\stu))$, $\stuj \wbtwn_{\matone(\stu)} \stu$).
	Throughout the TTC, each school $\sch$ is active if at least one student in $\matoneinv(\sch)$ is active (at the start of the TTC, every school such that $|\matoneinv(\sch)| \geq 1$ is active by construction).
	
	In every round of the TTC, for each student, construct the list of \textit{admissible} schools.
	A school $\sch$ is admissible to student $\stu$ if there is some active $\stuj \in \matoneinv(\sch)$ such that $\stu \ebtwn_\sch \stuj$.
	Each student points to her most-preferred admissible school, and each active school points to the highest-priority active student in $\matoneinv(\sch)$ (by definition, a school is only active if there is an active student in $\matoneinv(\sch)$, so this is well-defined).
	Because every active student points to an active school, and every active school points to an active student, a cycle must form.
	Every student on the cycle is assigned a seat at the school she points to and becomes inactive.
	The process is repeated until no students (and hence no schools) are active.
	\Cref{alg:TTC} provides the formal definition.
	
	\begin{algorithm}
		\caption{Top Trading Cycles Algorithm}\label{alg:TTC}
		\begin{algorithmic}
			\State initialize $\mattwo \gets \matone$
			\ForAll{$\stu \in \Stu$}  \Comment{activates students in the lowest-priority group at their assigned school}
				\State initialize $\stu$ as \textit{active} if $\stuj \wbtwn_{\matone(\stu)} \stu$ for every $\stuj \in \matoneinv(\matone(\stu))$
			\EndFor
			\For{$\sch \in \Sch$}
			\State initialize $\sch$ as \textit{active} if there is an active student in $\matoneinv(\sch)$
			\EndFor
			\While{there is an active student}
				\For{every active student $\stu$} \Comment{students point to most-preferred admissible school}
					\State $\stu$ points to her $\spref_\stu$-most preferred active school in $\{\sch  \spbr \t{exists active } \stuj\in \matoneinv(\sch) \st \stu \ebtwn_\sch \stuj\}$
				\EndFor
				\State each active school points to the highest-priority active student in $\matoneinv(\sch)$
				\State a cycle $\stu_1, \sch_1, \ldots, \stu_n, \sch_n$ exists
				\ForAll{$1\leq k\leq n$} \Comment{removes students in cycle from problem, then restarts}
					\State set $\mattwo(\stu_k)$ to $\sch_k$
					\State deactivate $\stu_k$
					\State if there are no active students in $\matoneinv(\sch_k)$, then deactivate $\sch$
				\EndFor
			\EndWhile
		\end{algorithmic}
	\end{algorithm}
%	\begin{algorithm}
%		\caption{Top Trading Cycles Algorithm}\label{alg:TTC}
%		\begin{algorithmic}
%			\State initialize $\mattwo \gets \matone$
%			\State for every $\sch \in \Sch$, initialize $\sch$'s \textit{active priority group} as the students who are in the same priority group as the \textit{lowest} priority student in $\matoneinv(\sch)$
%			\State initialize the \textit{active students} as the students who are in the active priority group of their $\matone$-school
%			\While{there is an active student}
%			\State each school points to the highest-priority active student in $\matoneinv(\sch)$, if any
%			\State each active student $\stu$ points to her most-preferred school $\sch$ such that $\stu$ is in the active priority group of $\sch$ and $\sch$ points to a student
%			\State a cycle $\stu_1, \sch_1, \ldots, \stu_n, \sch_n$ exists
%			\ForAll{$1\leq k\leq n$}
%			\State set $\mattwo(\stu_k)$ to $\sch_k$
%			\State remove $\stu_k$ from the active students
%			\State decrement the counter of $\sch_k$ by one
%			\EndFor
%			\EndWhile
%		\end{algorithmic}
%	\end{algorithm}
	
	\begin{theorem}\label{thm:wthn-block}
		$\mattwo$ is in the unified core.
	\end{theorem}
	The full proof of \Cref{thm:wthn-block} is available in the Appendix.
	Here, I outline the key points in the argument.
	A preliminary step is to show that $\mattwo$ is not between-group blocked; given that $\mattwo$ is a Pareto improvement of $\matone$ for the students but is equivalent under the between-priorities, this is not difficult to show.\footnote{But the importance of this cannot be overstated; as I discuss in \Cref{sec:disc}, having to find a Pareto improvement of $\matone$ is has restricted the kinds of algorithms available in school choice.}
	The rest of the proof supposes that $\Coal$ within-group blocks $\mattwo$ with $\nat$, and then finds a contradiction.
	
	The major difficulty with applying a standard proof that the TTC is in the core is that, for a school $\sch$ in $\nat(\Coal)$, not every student in $\natinv(\sch)$ is necessarily in $\Coal$.
	Put differently, using the TTC usually requires a well-defined set of owners, which the unified core lacks.
	The crux of the proof is a construction of a bijection $\Topt_\sch$ between students in $\natinv(\sch) \cap \Coal$ and students in $\mattwoinv(\sch) \cap \Coal$, which identifies an owner for each relevant seat.
	I construct $\Topt_\sch$ by addressing students in both $\natinv(\sch) \cap \Coal$ and $\mattwoinv(\sch) \cap \Coal$ separately from those in only $\natinv(\sch) \cap \Coal$.
	For those in both, I make $\Topt_\sch$ the identity function.
	For those in only $\natinv(\sch) \cap \Coal$, I leverage the $\natinv(\sch) \ebtwn_\sch \mattwoinv(\sch)$ of within-group enforcement requirement to find an equally-sized set of students in $(\mattwoinv(\sch) \cap \Coal )\backslash \natinv(\sch)$.
	I further show that if $\stu$ is active in the TTC, then $\Topt_{\nat(\stu)}(\stu)$ is also active in the TTC and $\Topt_{\nat(\stu)}(\stu) \in \Upp_{\nat(\stu)}(\stu)$.
	Hence, for active students, the $\Topt_\sch$ essentially identifies which student owns which seat at $\sch$.
		
	With the bijection $\Topt_\sch$ in hand, I then can use standard methods.
	I construct a cycle of students $\stu_1, \ldots \stu_n$ in $\Coal$ where $\nat(\stu_k) = \mattwo(\stu_{k+1})$ as in most proofs involving the TTC.
	I construct it such that least one student in this cycle must strictly prefer $\nat$ to $\mattwo$.
	But I also show that every $\stu_k$ is active in the Top Trading Cycles algorithm, so each must be deactivated only after the next student.
	This then leads to a situation in which $\stu_k$ must be deactivated strictly before $\stu_k$ is, a contradiction.

	\section{Discussion}\label{sec:disc}
	In this section, I place \Cref{thm:wthn-block} and the DA-TTC in conversation with existing results and point toward several avenues for future research.
	I first compare this result with that of EE, highlighting the complexities arising from the within-group blocks.
	I then turn to the question of finding a constrained efficient match.
	Finally, I turn to a more general discussion of the DA and TTC in school choice, and highlight the difficulties present in bridging these two algorithms.
	
	With the proof of the non-emptiness of the unified core in hand, I turn to a discussion of the differences between the DA-TTC and the Stable Improvement Cycles algorithm (DA-SIC) of EE.
	For context, EE considers the matches that are not between-group blocked.
	The DA-SIC starts by running the DA using the within-group priorities to establish a baseline match that is not between-group blocked.
	The SIC then checks for a cycle of students who each prefer the next student's match such that, if the students exchange seats, no between-group blocks are created; this is the eponymous \textit{stable} improvement cycle.
	The output of the SIC stage is a student-optimal match among those that are not between-group blocked.
	The critical difference between the SIC and the TTC is that the SIC allows for trades \textit{across} priority groups while the TTC does not.
	The following example with three students and two schools (each with one seat) illustrates this:
		\begin{multicols}{4}
			\centering
			\begin{tabular}{C | C | C}
				\spref_1 & \spref_2 & \spref_3 \\
				\hline
				a & b & a \\
				b & a & b \\
			\end{tabular}
			\\
			\begin{tabular}{C | C }
				\sbtwn_a & \sbtwn_b \\
				\hline
				2 & 1   \\
				1,3 & 2,3  \\
			\end{tabular}
			\\
			\begin{tabular}{C | C   }
				\swthn_a & \swthn_b \\
				\hline
				2 & 1   \\
				3 & 3  \\
				1 & 2  \\
			\end{tabular}
			\\
			\begin{align*}
				\mat^{\t{DA}}&= 
				\begin{pmatrix}
					1 & 2 & 3  \\
					b & a & 3
				\end{pmatrix}
			\end{align*}
		\end{multicols}
	At the match $\mat^{\t{DA}}$, EE identifies a stable improvement cycle of $1 \rightarrow a$ and $ 2 \rightarrow b $.
	However, allowing this pair to trade schools violates student $3$'s within-priority at both schools: after the trade, the coalition of $\{3\}$ can within enforce a match to either $a$ or $b$.
	In effect, the stable improvement cycle is stable with respect to $\wbtwn$, but disregards $\swthn$.\footnote{There is a second difference between the DA-TTC and the DA-SIC. Consider a school choice problem with one school (with one seat) and two students, both of whom desire the seat at the school. If the students are in the same priority group, then the DA-SIC could assign either student to the school. The DA-TTC, however, is constrained to only assign the higher-priority student. Put differently, the DA-TTC can emulate the TTC because the DA-TTC makes use of the within-group priority; the DA-SIC cannot.}
	This allows for these trades between priority groups.
	This difference is not a mere technical nuance; in the case of BPS, it is precisely trades like this that concern the policymaker.
	
	Thus far, I have dealt only with existence rather than constrained-efficiency.
	The major drive of EE is to find a student-optimal match subject to the stability constraints, which they show can be found by iteratively eliminating stable improvement cycles.
	Given the similarity between the DA-TTC and the DA-SIC, the natural question arises as to whether an analogous method can be used to eliminate \say{unified} stable improvement cycles.
	Unfortunately, this approach seems doomed to failure.
	Although for the between-priority the stable improvement cycle can be imported as-is, checking whether a cycle respects within-priority is significantly more complex.
	The crux is that a student's within-priority can be violated when a higher-priority student trades her seat to a lower-priority agent.
	However, in a cycle, whether such a \say{trade} can be found depends upon the matches of students outside the cycle.
	Put differently, a stable improvement cycle may involve trades across priority groups, but then a trade {within} the priority group must be found to protect the former trade.
	The problem is compounded by the possibility that the latter trade may itself be the result of another simultaneous trade.
	Whether a constrained efficient match can be found using this technique or others is an open question.
	
	There are several open questions about algorithms similar to the DA-TTC.
	The DA-TTC is a part of a growing set of algorithms that can be roughly described as \say{DA baseline, then trade to improve.}
	Starting with \cite{abdulkadiroglu_school_2003}, several authors have used this formula, such as EE, \cite{kesten_school_2010}, and \cite{doe_matching_2024}.
	The difficulty with combining these two families of algorithms is maintaining monotonicity on a set of matches.
	\cite{echenique_theory_2004} shows how the DA can be viewed as a monotonic function on matches; however, the TTC does not share this property \citep{echenique_stable_2024}.
	When combining these two algorithms, care must be taken to guarantee that the trades made in the TTC do not upset the match found in the DA, which is why the previous papers allow students only to trade the school they received from the DA rather than any school they have the highest-priority for.
	Going the other way around is tougher: if the designer first allows students to trade schools and then attempts to remove fairness violations, a student may trade for a school that is taken away; with one student-school pair removed from the group of trading students, the outcome of the TTC shifts unpredictably.
	The only paper to tackle both simultaneously is \cite{abdulkadirog_generalized_2011}.
	Is this approach to connecting the DA and TTC the only viable option?
	Can conditions be placed upon priorities to ameliorate the non-monotonicity of the TTC in the presence of the DA?
	
	Another obstacle to implementing the TTC is the complexity of the algorithm.
	Payzant alludes to this; after his previous quote, he continued: \say{Moreover, Top Trading Cycles is less transparent -- and therefore more difficult to explain to parents -- because of the trading feature executed by the algorithm, which may perpetuate the need or perceived need to \say{game the system.}}
	Concerns like Payzant's have spurred a growing body of literature to understand the complexity of the TTC \citep{gonczarowski_structural_2024, leshno_simple_2017}.
	In this paper, I instead focused on properties of the match rather than the process used to reach the match.
	The properties of a match may prove more transparent to stakeholders than the properties of an algorithm.

	I close with a brief note about the DA and TTC.
	Within the literature on school choice, the vast majority of algorithms stem from the DA and the TTC.
	The monotonicity of the DA is well-suited for the lattice structure of the stable matches, and the TTC flexibly handles ownership economies \citep{gusfield1989stable,papai_strategyproof_2000}.
	However, as I note in the previous paragraphs, these algorithms are difficult to combine.
	Additionally, given the complexity of the TTC and how few districts have attempted to implement it, relying on TTC-style algorithms to improve the efficiency of the DA may not be successful in applications.
	Developing general algorithms---as in \cite{abdulkadirog_generalized_2011}---to probe the efficiency-fairness frontier is a research direction of primary importance.

	\bibliography{sampleOutline}

	\appendix
	\setcounter{lemma}{0}
	\renewcommand{\thelemma}{\Alph{section}.\arabic{lemma}}
	\setcounter{definition}{0}
	\renewcommand{\thedefinition}{\Alph{section}.\arabic{definition}}
	\section{Omitted Proofs}\label{appendix}
	
	\begin{lemma}\label{lem:IR}
		Every student weakly prefers $\matone$ to being unmatched.
	\end{lemma}
	\begin{proof}
		Observe that no student rejects a proposal from herself (because students do not make rejections).
		This implies that no student prefers being unmatched more than $\matone$.
	\end{proof}
	\begin{proof}[\unskip\nopunct]
		\textbf{\textit{Proof of \Cref{lem:btwn-block}:}}
		
		Suppose (toward a contradiction) that $\matone$ is between-group blocked by a coalition $\Coal$ between-enforcing a match $\nat$ which it prefers.
		
		First, by \cref{lem:IR}, there is some school in $\matone(\Coal)$.
		
		Second, observe that a school only rejects students if it receives more proposals (cumulatively) than its capacity $\quo_\sch$.
		Therefore, if $\stu$ strictly prefers a school to $\matone(\stu)$, then that school is filled to capacity at $\matone$.
		
		The previous two points imply the existence of a school $\sch$ and two students $\stu$ and $\stuj$ with the following properties:
		\begin{enumerate}
			\item $\sch \in \nat(\Coal)$; and
			\item $\stu \in \Coal \cap \natinv(\sch) $ but $\stu \notin \matoneinv(\sch)$; and
			\item $\stuj\in \matoneinv(\sch)$ but $\stuj\notin \natinv(\sch)$; and 
			\item $\stu \sbtwn_\sch\stuj$.
		\end{enumerate}
		But then notice that $\stu$ must have proposed to $\sch$ but have been rejected.
		However, $\stuj$ must have proposed to $\sch$ but \textit{not} have been rejected.
		If $\stuj $ proposed to $\sch$ before $\sch$ rejects $\stu$, then $\sch$ rejects $\stuj$ before rejecting $\stu$, a contradiction.
		If $\stuj$ proposed to $\sch$ after $\sch$ rejects $\stu$, then $\sch$ must have rejected $\stu$ in favor of a higher within-group priority student $\stu^*$; hence, $\stuj$ is also rejected, a contradiction.
		
		Therefore, there is no coalition which between-group blocks $\matone$.
	\end{proof}
	\begin{lemma}\label{lem:ebtwn}
		If $\Coal \ebtwn_\sch \Coal' $, then for every $\stu\in \Stu$: $|[\stu]_\sch \cap \Coal | = |[\stu]_\sch \cap \Coal ' |$ 
	\end{lemma}
	\begin{proof}
		Suppose (toward a contradiction) that for some $\stu\in \Stu$: $|[\stu]_\sch \cap \Coal | \neq |[\stu]_\sch \cap \Coal ' |$.
		Without loss of generality, let $\stu$ be in the $\wbtwn$-highest priority group such that $|[\stu]_\sch \cap \Coal | \neq |[\stu]_\sch \cap \Coal ' |$.
		Again, without loss of generality, let $|[\stu]_\sch \cap \Coal | > |[\stu]_\sch \cap \Coal ' |$.
		Consider that $\Coal \ebtwn_\sch \Coal' $ implies $\Coal '  \wbtwn_\sch \Coal$.
		But then $\Coal ' $ and $\Coal$ cannot be indexed such that $\coal_k ' \wbtwn_\sch \coal_k $ because every student in a $\wbtwn$-more preferred priority group must be paired with a student in that same priority group and the students in the $[\stu]_\sch$ priority group are imbalanced.
		This is a contradiction.
	\end{proof}
	\begin{lemma}\label{lem:trans-block}
		The following statements are true:
		\begin{enumerate}
			\item 
			If $\Coal$ within-group blocks $\mattwo$ by within-enforcing $\nat$, then $\Coal$ within-group blocks $\matone$ by within-enforcing $\nat$ as well.
			\item 
			If $\Coal$ between-group blocks $\mattwo$ by between-enforcing $\nat$, then $\Coal$ within-group blocks $\matone$ by within-enforcing $\nat$ as well.
		\end{enumerate}
	\end{lemma}
	\begin{proof}
		I show that $\Coal$ can within-group enforce or between-group enforce $\nat$ over $\matone$ (respectively).
		Because every student weakly prefers $\mattwo$ to $\matone$, it follows that $\Coal$ prefers $\nat$ to $\matone$ and hence within-group blocks or between-group blocks $\matone$ with $\nat$.
		
		For the first statement, suppose that $\Coal$ within-group blocks $\mattwo$ by within-enforcing $\nat$.
		I show that $\Coal$ can within-group enforce $\nat$ over $\matone$.
		To see this, notice that $\Upp_\sch(\stu)$ is independent of $\mattwo$.
		Hence, the sum of within-group interrupters across all students in $\natinv(\sch) \cap \Coal$ is the same for both $\mattwo$ and $\matone$.
		Additionally, by definition, for every $\sch\in \nat(\Coal)$: $\mattwoinv(\sch) \ebtwn_\sch \natinv(\sch)$.
		By the construction of $\mattwo$: $\mattwoinv(\sch) \ebtwn_\sch \matoneinv(\sch)$.
		Hence, $\nat(\sch) \ebtwn_\sch \matoneinv(\sch)$.
		Therefore, $\Coal$ can within-group enforce $\nat$ over $\matone$.
		
		For the second statement, suppose that $\Coal$ between-group blocks $\mattwo$ by within-enforcing $\nat$.
		I show that $\Coal$ can between-group enforce $\nat$ over $\matone$.
		To see this, notice that for every school $\sch$, $\mattwo \ebtwn_\sch \matone$.
		Hence, if $\natinv(\sch) \sbtwn_\sch \mattwo$, then $\natinv(\sch)\sbtwn_\sch \matone$.
		Therefore, $\Coal$ can between-group enforce $\nat$ over $\matone$.
		
		This completes the proof.
	\end{proof}
	\begin{lemma}\label{lem:capacity}
		If $\Coal$ within-group blocks $\mattwo$ with $\nat$, then $\stu \in \Coal$ strictly prefers $\nat(\stu)$ to $\mattwo(\sch)$ only if $\sch$ is filled to capacity at $\mattwo$.
	\end{lemma}
	\begin{proof}
		Suppose (toward a contradiction) that $\stu\in\Coal$ strictly prefers $\nat(\stu) \equiv \stu$ to $\mattwo(\sch)$ but $|\mattwoinv(\sch)| < \quo_\sch$.
		
		By \cref{lem:trans-block}, $\Coal$ within-group blocks $\matone$ with $\nat$.
		By the construction of $\mattwo$, $|\mattwoinv(\sch)| = |\matoneinv(\sch)|$.
		By the construction of $\matone$, $\sch$ only rejects $\stu$ if more than $|\quo_\sch|$ students point to $\sch$.
		But less than $\quo_\sch$ other students point at $\sch$, so $\sch$ does not reject $\stu$.
		Therefore, $\stu$ weakly prefers $\matone $ to $\stu$.
		This contradicts that $\stu$ strictly prefers $\sch$ to $\mattwo$.
		Therefore, $|\mattwoinv(\sch)| = \quo_\sch$.
	\end{proof}
	\begin{proof}[\unskip\nopunct]
		\textbf{\textit{Proof of \Cref{thm:wthn-block}:}}
		By \cref{lem:trans-block}, $\mattwo$ is not between-group blocked; otherwise, $\matone$ would be between blocked, a contradiction to \cref{lem:btwn-block}.
		
		Now consider within-group blocks.
		Suppose (toward a contradiction) that $\mattwo$ is within-group blocked by a coalition $\Coal$ within-enforcing a match $\nat$.
		
		\textbf{First}, for every school $\sch \in \nat(\Coal)$, I define an injection $\Topt_\sch: \natinv(\sch) \cap \Coal \rightarrow \matoneinv(\sch) \cap \Coal$ in the following manner.
		Let $\natinv(\sch) \cap \Coal = \{\stu_1, \stu_2, \dots, \stu_n\}$ be indexed such that $\stu_{k+1} \swthn_\sch \stu_{k}$ for every $k$.
		Notice that by the construction of $\matone$, if $\stu_k \in \matoneinv(\sch)$, then for $k ' > k$, $\stu_{k'} \in \matoneinv(\sch)$.
		Therefore, there is a unique $m$ (possibly taking the value of $0$ or $n$) such that $\stu_{m} \notin \matoneinv(\sch)$ but $\stu_{m + 1} \in \matoneinv(\sch)$.
		I define $\Topt_\sch$ piecewise based on the index $k$:
		\begin{itemize}
			\item For $m+1 \leq k \leq n$,
			let $\Topt_\sch(\stu_k) \equiv \stu_k$.
			
			Observe that $\Topt_\sch(\stu_k) \in \matoneinv(\sch)$ by the definition of $m$.
			Notice that $\Topt_\sch(\stu_k) \in \Coal$ because $\stu_k \in \natinv(\stu) \cap \Coal$.
			This piece of $\Topt_\sch$ is clearly injective.
			
			\item For $1 \leq k \leq m$,
			consider the following argument:
			
			Because $\Coal$ within-group enforces $\nat$ at $\mattwo$, it follows that $\natinv(\sch) \ebtwn_\sch \mattwoinv(\sch)$.
			By construction, $\matoneinv(\sch) \ebtwn_\sch \mattwoinv(\sch)$.
			By construction of $\matone$, every student in $\matoneinv(\sch)$ is strictly $\swthn_\sch$-preferred to every student in $\stu_1, \ldots, \stu_{m}$. 
			Hence by \cref{lem:ebtwn}, $\stu_1, \ldots, \stu_{m} \in [\stu_{m}]_\sch$.
			
			Again, because $\Coal$ within-group enforces $\nat$ over $\mattwo$, there are at most $\quo_\sch - |\natinv(\sch)|$ within-group interrupters for student $\stu_{m}$.
			Because $|\natinv(\sch)| = \quo_\sch$ by \cref{lem:capacity}, it follows that there are no within-group interrupters of $\stu_{m}$.
			Hence, $\Upp_\sch(\stu_m) \subseteq \Coal$.
			But every student in $\matoneinv(\sch)$ is $\swthn_\sch$-preferred to $\stu_m$.
			Thus, $\matoneinv(\sch) \cap [\stu_m]_\sch \subseteq \Upp_\sch(\stu_m)$.
			Therefore, $\matoneinv(\sch) \cap [\stu_m]_\sch \subseteq \Coal$.
			
			Finally, note that $|\matoneinv(\sch) \cap [\stu_m]_\sch | = |\natinv(\sch) \cap [\stu_m]_\sch|$ by \cref{lem:ebtwn}.
			Let $M$ be the greatest index such that $\stu_M \in [\stu_m]_\sch$, and note that $\{\stu_1, \ldots \stu_M\}\subseteq\natinv(\sch) \cap [\stu_m]_\sch$.
			But then I have
			\begin{align*}
				\Big|\matoneinv(\sch) \cap [\stu_m]_\sch \Big| \geq \Big|\{\stu_1, \ldots, \stu_{m}, \ldots \stu_M\} \Big|
			\end{align*}
			However, because $\{\stu_{m+1}, \ldots , \stu_{M}\} \subseteq \matoneinv(\sch)$, this can be rewritten:
			\begin{align*}
				\Big|\big(\matoneinv(\sch) \cap [\stu_m]_\sch\big) \backslash \{\stu_{m+1}, \ldots , \stu_{M}\}\Big| + \Big|\{\stu_{m+1}, \ldots , \stu_{M}\} \Big| \geq \Big|\{\stu_1, \ldots, \stu_{m}, \ldots \stu_M\} \Big|
			\end{align*}
			This further implies
			\begin{align}\label{equ:inject}
				\Big|\big(\matoneinv(\sch) \cap [\stu_m]_\sch\big) \backslash \{\stu_{m+1}, \ldots , \stu_{M}\}\Big| \geq \Big|\{\stu_1, \ldots, \stu_{m}\} \Big|
			\end{align}
			
			Define $\Topt_\sch(\stu_k) $ as the $k^\t{th}$ $\swthn_\sch$ least-preferred student in $\big(\matoneinv(\sch) \cap [\stu_m]_\sch\big) \backslash \{\stu_{m+1}, \ldots , \stu_{M}\}$.
			By \cref{equ:inject}, $\Topt_\sch(\stu_k)$ is well-defined and is an injection on this piece.
			Because $\Topt_\sch(\stu_k) \notin \{\stu_{m+1}, \ldots , \stu_{M}\}$, this piece of $\Topt_\sch$ has no overlap with the first piece.
			Finally, by the above argument, $\Topt_\sch(\stu_k) \in \Coal$.
		\end{itemize}
		Hence, there is a well-defined injection $\Topt_\sch: \natinv(\sch) \cap \Coal \rightarrow \matoneinv(\sch) \cap \Coal$.
		
		\textbf{Second}, iteratively construct the following sequence of students.
		Let $\stu_1$ be some student in $\Coal$ such that $\nat(\stu_1) \neq \mattwo(\stu_1)$.
		Because $\Coal$ prefers $\nat$ to $\mattwo$, such a student exists.
		Let $\stu_{k+1} \equiv \Topt_{\nat(\stu_{k})}(\stu_{k})$.
		
		\textbf{Third},	I claim that $\stu_{k}$ is well-defined, $\stu_{k} \in \Coal$, and $\matone(\stu_k) \neq \nat(\stu_k)$.
		I show this by induction.
		For the base case when $k = 1$, I make three observations:
		\begin{enumerate}
			\item That $\stu_1$ is well-defined is noted previously.
			\item $\stu_1 \in \Coal$ by definition.
			\item Because $\stu_1$ strictly prefers  $\nat(\stu_1)$ to $\mattwo(\stu_1)$ and $\stu_1$ weakly prefers $\mattwo$ to $\matone$ by construction, it follows that $\nat(\stu_1) \neq \matone(\stu_1)$.
		\end{enumerate}
		For the inductive step, suppose that every student with index less than $k$ is well-defined, is in $\Coal$, and is not matched to the same school in $\nat$ and $\matone$.
		I make three observations:
		\begin{enumerate}
			\item Note that $\nat(\stu_{k-1}) \in \nat(\Coal)$.
			Similarly, note that $\nat(\stu_{k-1}) \in \Sch$ because $\mattwo$ is individually rational.
			Hence, $\Topt_{\nat(\stu_{k-1})}$ is well-defined.
			Therefore $\stu_k$ is well-defined.
			\item $\stu_k \in \Coal$ because the domain of $\Topt_{\nat(\stu_{k-1})}$ is a subset of $\Coal$.
			\item Because $\nat(\stu_{k-1})\neq \mattwo(\stu_{k-1})$ and $\stu_{k-1}$ weakly prefers $\mattwo$ to $\matone$, it follows that $\nat(\stu_{k-1}) \neq \matone(\stu_{k-1})$.
			Thus, $\stu_{k-1} \notin \matoneinv(\nat(\stu_{k-1}))$.
			By the construction of $\Topt_{\nat(\stu_{k-1})}$, it follows that $\matone(\stu_k) = \nat(\stu_{k-1})$ and $\nat(\stu_k)\neq \nat(\stu_{k-1})$.
			Therefore, $\matone(\stu_k) \neq \nat(\stu_k)$.
		\end{enumerate}
		
		\textbf{Fourth}, because each $\Topt_{\nat(\stu_k)}$ is an injection and each student $\stu_k$ is matched to only one school in $\matone$, it follows that if $\stu_k \neq \stu_{l}$, then $\Topt_{\nat(\stu_k)} \neq \Topt_{\nat(\stu_l)}$.
		Because there are a finite number of students, it follows that there is a (minimal) index $n$ such that $\stu_1 = \stu_{n+1}$.
		
		\textbf{Fifth}, consider the cycle $\stu_1, \ldots, \stu_n$.
		Because $\Topt_{\nat(\stu_{k-1})}(\stu_{k-1})\neq \stu_k$, it follows that $\stuj \wbtwn_{\matone(\stu_k)}\stu_k$ for every $\stuj \in \matoneinv(\matone(\stu_k))$.
		Hence, every student in the cycle is active the Top Trading Cycles stage.
		Notice that each $\stu_k$ must be deactivated weakly later than $\stu_{k+1}$ because $\stu_k$ points to $\matone(\stu_{k+1})$.
		However, $\stu_1$ must be deactivated \textit{strictly} later than $\stu_2$ because $\nat(\stu_1)\neq \matone(\stu_2)$.
		Therefore, $\stu_1$ must be deactivated strictly before $\stu_1$, a contradiction.
		
		\textbf{Therefore}, there is no coalition that within-group blocks $\nat$.
	\end{proof}

%\\%-%-%-%-%-%-%-%-%-%-%-%-%-%-%-%-%-%-%-%-%-%-%-%-%-%-%-%-
\end{document}